\documentclass[12pt]{amsart}
\usepackage{graphicx} 
\DeclareGraphicsExtensions{.pdf,.eps,.fig}
\usepackage{color}
\usepackage[dvipsnames]{xcolor}
\usepackage{amsmath,todonotes}
\usepackage{amssymb}
\usepackage{amsthm}
\usepackage{amscd}
\usepackage{amsfonts}
\usepackage{fancyhdr}
 \usepackage{hyperref}
 \usepackage{float}
\usepackage[symbol]{footmisc}

\DeclareGraphicsExtensions{.pdf,.eps,.fig}

\newcommand{\trho}{\widetilde{\rho}}
\newcommand{\fr}[1]{\mathfrak{#1}}
\newcommand{\qsp}[2]{\,\ensuremath{\raise.5ex\hbox{$#1$}\big\slash\raise-.5ex\hbox{$#2$}}} 
\newcommand{\orb}[1]{\mathcal{O}^{(#1)}}
\newcommand{\stab}[1]{\mathfrak{h}^{(#1)}}
\newcommand{\nor}[1]{\mathfrak{n}^{(#1)}{}}
\newcommand{\nablat}[1]{\mathfrak{N}^{(#1)}}

\theoremstyle{plain} \numberwithin{equation}{section}
\newtheorem{theorem}{Theorem}[section]
\newtheorem{corollary}[theorem]{Corollary}

\newtheorem{lemma}[theorem]{Lemma}
\newtheorem{proposition}[theorem]{Proposition}
\newtheorem{propdef}[theorem]{Proposition-Definition}
\theoremstyle{definition}
\newtheorem{definition}[theorem]{Definition}

\newtheorem{remark}[theorem]{Remark}
\newtheorem{ex}[theorem]{Example}
\title{K\"ahler fibrations in quantum information theory}
\author{Ivan Contreras}
\address{Department of Mathematics, Amherst College, Massachussets, USA.}
\email{icontreraspalacios@amherst.edu}
\author{Michele Schiavina}
\address{Institute for Theoretical Physics, ETH Z\"urich, Wolfgang Pauli strasse 27, 8092, Z\"urich, Switzerland}
\address{Department of Mathematics, ETH Z\"urich, R\"amistrasse 101, 8092, Z\"urich, Switzerland}
\email{micschia@phys.ethz.ch}

\begin{document}

\begin{abstract}
    We discuss the fibre bundle of co-adjoint orbits of compact Lie groups, and show how it admits a compatible K\"ahler structure. The case of the unitary group allows us to reformulate the geometric framework of quantum information theory. In particular, we show that the Fisher information tensor gives rise to a structure that is sufficiently close to a K\"ahler structure to generalise some classical result on co-adjoint orbits.
\end{abstract}
\maketitle

\tableofcontents

\section*{Introduction}
The geometric description of finite dimensional quantum mechanics is an application of the Lie theory of the unitary group that gives rise to a number of nontrivial insights on a fairly established area of symplectic geometry. As argued in \cite{Chat} and by the authors in \cite{CES,ES1,ES2}, the orbit theory for compact Lie groups provides a natural environment to reformulate and answer questions in quantum mechanics and information theory. 

The spaces of pure and mixed quantum states are linked to the (co-)adjoint orbits of the unitary group in the same way complex projective spaces are a model for pure states, i.e. projectors, and closed systems evolve unitarily,  so that conservative dynamics is phrased in terms of the (co-)adjoint action. The natural geometric structures that exist on such orbits have a  direct interpretation in quantum mechanics and information theory, e.g. the Berry phase is related to the symplectic structure on orbits. We argue that the reverse is also true, i.e. the symplectic and K\"ahler geometry on the orbits can be obtained by analyzing the behavior of physical systems.

In \cite{CES} we began the study of the quantum Fisher information tensor (FIT) - a multi-parameter generalisation of the Fisher information index \cite{ES1,ES2,Chat} - and we clarified its relationship with both the Konstant-Kirillov-Souriau symplectic form $\Omega^{(\eta_0)}_{KKS}$ on a coadjoint orbit $\orb{\eta_0}$, and with the quantum Fisher information metric. The definition of FIT is borrowed from physics and has no \emph{a priori} analogue in terms of abstract Lie theory and symplectic geometry. 

The motivation that lead to it was to find a way to properly define and compute the Fisher information index in full generality\footnote{Before only partial answers were known and only for very particular cases \cite{B-NG,Luati1,Luati2}.}, to then proceed to solve the Fisher information optimisation problem for quantum measurements. Since the new redefinition of the FIT is natural from the point of view of the orbit theory, we argue that it might also play a role in a more abstract geometric sense.

In this note we review the theory of symplectic fibrations and extend  this construction to a certain strict class of K\"ahler fibrations where the total space is also K\"ahler. This will allow us to describe the geometry of quantum states in terms of strong K\"ahler fibrations of coadjoint orbits, with respect to which the Fisher information tensor is compatible in a natural way. 

We prove that the FIT is a K\"ahler structure only on degenerate orbits, but it is \emph{not far} from one. This suggests the definition of a more general geometric structure than  K\"ahler - we call it a Fisher-K\"ahler structure - where the compatibility between a Riemannian metric and a symplectic form is ensured only up to diagonal rescalings; one could say in a \emph{multi-projective} fashion. On degenerate orbits, i.e. complex projective spaces, it reduces to (a multiple of) the Fubini study metric; we think of the FIT as its generalisation in this sense.

This structure is likely relevant in the context of geometric quantisation, where the integrality of symplectic structures is crucial and where one might thus be able to distinguish between structures by providing another nontrivial invariant. We argue that the alternative symplectic form coming from the Fisher tensor contruction might give rise to different pre-quantisations, the study of which - although not be presented here - might lead to new insights for the geometric quantisation of co-adjoint orbits.

The paper is organised as follows. In Section \ref{Sect:sympkahl} we define the relevant geometric background we will use throughout, and introduce the notion of Symplectic and K\"ahler fibrations. We will use a stricter version of K\"ahler fibration than done elsewhere, which will be directly adaptable to the main case addressed in this paper. In this section we also define the notion of Fisher structure as a scaling-dependent generalisation of complex structures and analyse its linear model.

In Section \ref{Section:Fibration} we turn to the theory of co-adjoint orbit of compact Lie groups, showing how they provide a natural example of strong K\"ahler fibrations. We will introduce the main constructions that will be needed in the rest of the paper.

Section \ref{Section:Orbits} is devoted to the case $G=U(n)$, which will model quantum mechanics and the spaces of states. We will show how fibrations of co-adjoint orbits can be seen as a nesting of spaces of mixed states and how the formalism allows us to understand the Fisher information tensor on such fibrations.

\subsection*{Acknowledgments}
M.S. acknowledges partial support by SNF grant No. P2ZHP2$\_$164999.  I.C.
acknowledges partial support by SNF Grant P300P2$\-$154552. We thank the anonymous referee for valuable comments that improved the quality and presentation of the paper, in particular, the distinction between the real and complexified version of the Lie algebra in Theorem \ref{TH:compatiblecomplexstructure}.

\section{Generalities of symplectic and K\"ahler fibrations}\label{Sect:sympkahl}
In this section we review the construction of fibre bundles and connections in the symplectic and K\"ahler framework. The construction of symplectic fibrations reviewed in here can be found, for instance, in \cite{Guillemin}. Some of the definitions and results regarding K\"ahler fibrations are borrowed from \cite{Antonelli}, and some others are new. 

\subsection{Symplectic and K\"ahler Fibrations}
\begin{definition}(Symplectic fibre bundle)\label{symfib}. A fibre bundle $\pi:M\to B$ is a \emph{symplectic fibre bundle} if the following conditions are satisfied:
\begin{enumerate}
\item The fibres are symplectic manifolds $(F_b,\omega_{F_b}),\ \forall b\in B$.
\item The transition maps:
$$\phi_{\alpha,\beta}: (\mathcal U_{\alpha}\cap \mathcal U_{\beta})\times F \to (\mathcal U_{\alpha}\cap \mathcal U_{\beta})\times F$$
are symplectomorphisms of the fibers.
\end{enumerate}
\end{definition}
A natural question to ask is whether there exists a two-form $\omega$ on the total space $M$ that restricts to the symplectic forms on fibres. This suggests the following
\begin{definition}\label{compconn} A $2$-form $\omega \in \Omega^2(M)$ on a symplectic fibre bundle $F\hookrightarrow M\to B$  is called \emph{fibre-compatible} if the following condition is satisfied:
$$\omega\vert_F\cong \iota_F^*(\omega)=\omega_F,$$
where $\iota_F$ denotes the inclusion of any fibre $F$ into $M$. 
\end{definition}

\begin{definition}(Strong symplectic fibre bundle).
A symplectic fibre bundle is called \emph{strong} if it admits a fibre-compatible form that is also \emph{symplectic}.
\end{definition}

Given a general section of the tensor bundle of the total space of a fibration with connection $(F,M,B,\nabla)$, we will ask that it be compatible with the horizontal distribution defined by the connection. This will turn out relevant for the main constructions in this paper.

\begin{definition}[$\nabla$-compatibility of tensors]\label{nablaCompatible}
Let $(F,M,B, \nabla)$
 be a fibre bundle with connection $\nabla$ and let  $A \in \mathbb{T}^{(k,l)}M:=\Gamma(TM^{\otimes k} \otimes T^*M^{\otimes l})$ be a $(k,l)$-tensor on $M$. Then, $A$ is said to be \emph{$\nabla$-compatible} if $A \vert_{F}$ is a tensor on $\mathbb{T}^{(k,l)}F$ and the connection $\nabla$ induces a splitting of $A$:
 \[A= A\vert_F \oplus A\vert_\nabla\]
 where $A\vert_\nabla$ is the tensor restricted to the horizontal distribution induced by $\nabla$.
 \em\end{definition}

\begin{definition}
Let $M\to B$ be a fiber bundle with typical fiber $F$. If $M$ admits a complex structure $J$, it will be called \emph{fibre-compatible} if its restriction to the fiber $J\vert_F$ is a complex structure on the fibres. 
\end{definition}

\begin{lemma}\label{compatibleComplexlemma}
Let $J\in\mathbb{T}^{(1,1)}(M)$ be a complex structure on the total space $M$ of a fiber bundle with connection $(F,M,B,\nabla)$. $J$ is fibre-compatible iff it is $\nabla$-compatible and $\mathrm{dim}(B)=2p$. Moreover, $J\vert_\nabla$ is a linear automorphism of the horizontal distribution defined by $\nabla$, such that $J\vert_\nabla^2=-\mathbb{I}$.
\end{lemma}

\begin{proof}
Assume $J$ fibre compatible. The connection $\nabla$ defines a splitting of $TM\simeq TF \oplus \mathrm{Hor}_\nabla$. Above each point $m\in M$ one can write $J$ in block form as
$$
J=\left(\begin{array}{cc}
J\vert_F & K \\
-K & Z
\end{array}\right)$$
but since $J^2=-\mathbb{I}$ and $J_F\vert_F^2 = -\mathbb{I}$ we conclude that necessarily $K=0$ and $Z^2 = - \mathbb{I}$.

Let now $J$ be $\nabla$-compatible, i.e. $J$ is block diagonal $J=\mathrm{diag}\{J\vert_F, J\vert_\nabla\}$. Then from $J^2=-\mathbb{I}$, knowing that $B$ (and therefore $F$) is even-dimensional we immediately prove the statement.
\end{proof}

Now we recall the notion of \emph{symplectic connection} \cite{Guillemin}, which encompasses a more restrictive compatibility condition between the symplectic structure on fibres and the parallel transport along the connection on the fibre bundle.

\begin{definition}(Symplectic connection). Given a connection $\nabla$ on the symplectic fibre bundle $F\hookrightarrow M\twoheadrightarrow B$, i.e. a splitting
$$TM\simeq {\mbox{Ver}}\oplus \mbox{Hor}_{\nabla},$$
it will be called \emph{symplectic} if the parallel transport along $\nabla$ of every loop $\gamma: S^1\to B$ is a symplectomorphism of the fibers. In other words: $\mbox{Hol}_{\nabla}(\gamma)_b$ is a subgroup of the symplectomorphism group $\mbox{Symp}(F_b)$.
\end{definition}

The following definition and subsequent theorems (see \cite{Guillemin} for more details) clarify this picture.

\begin{definition}\label{omegaCompatible}  Given a symplectic fibre bundle $M\to B$ with fibre-compatible two-form $\omega\in\Omega^2(M)$, a connection $\nabla$ is called $\omega$\emph{-compatible} if the following condition holds:
$$\mbox{Hor}_\nabla= \mbox{Ver}^{\omega},$$
where $\mbox{Ver}^{\omega}$ denotes the symplectic orthogonal space to $\mbox{Ver}$.
\end{definition} 

The following observation connects Definitions \ref{nablaCompatible} and \ref{omegaCompatible}.

\begin{lemma}
Let $\omega$ be a fibre-compatible two-form on the symplectic fibre bundle $M\to B$. A connection $\nabla$ is $\omega$-compatible if and only if $\omega$ is $\nabla$-compatible.
\end{lemma}

\begin{proof}
Consider an $\omega$-compatible connection $\nabla$, i.e. such that $\mathrm{Hor}_\nabla=\mathrm{Ver}^\omega$. Then $\omega$ naturally splits as $\omega\vert_F \oplus \omega\vert_{\nabla}$ since $\omega(v,w)=0$ whenever $v$ is vertical and $w$ horizontal.

On the other hand, if $\omega$ is $\nabla$-compatible it splits, and every vertical vector field will annihilate all horizontal ones. Since $\omega$ is fibre compatible, it is nondegenerate on $\mathrm{Ver}$, showing $\omega$-orthogonality of the horizontal and vertical subspaces, and thus $\omega$-compatibility of the connection.
\end{proof}

We have the following fundamental result relating symplectic fibrations and symplectic connections.

\begin{theorem}\cite{Guillemin}\label{SymConn}
Let $\pi: M\to B$ be a symplectic fibration.
\begin{enumerate}
\item If it comes equipped with a fibre-compatible form $\omega$, then this defines an $\omega$-compatible \emph{symplectic} connection $\nabla_{\omega}$ if and only if $\iota_{v_1\wedge v_2}d\omega =0$ for any two vertical vector fields $v_i\in\Gamma(\mathrm{Ver})$. 
\item If the fibres are compact, connected and simply connected and if it comes equipped with a symplectic connection $\nabla$, then there exists a unique fibre compatible, $\nabla$-compatible closed  two-form $\omega_\nabla$ that satisfies the condition
\begin{equation}
    \pi_*\omega_\nabla^{\mathrm{dim}(F)/2+1}=0.
\end{equation}
Moreover, any fibre-compatible, $\nabla$-compatible closed form $\widetilde{\omega}$ is of the form $\widetilde{\omega} = \omega_\nabla + \pi^*\omega_B$ for a given closed form $\omega_B\in\Omega^2_{cl}(B)$.
\end{enumerate}

\end{theorem} 

\begin{proof}[Sketch of Proof]We have that
\begin{enumerate}
    \item Define the horizontal subbundle $\mathrm{Hor}\subset TM$ as the symplectic orthogonal of the vertical subbundle $\mathrm{Ver}\subset TM$ with respect to the fibre-compatible form $\omega$, namely $\mathrm{Hor}:= \mathrm{Ver}^\omega$.  Then one shows that a Horizontal lift $\hat w$ of any vector field on $B$ acts by symplectomorphism on $\omega$ modulo $B$ if and only if $\iota_{v_1\wedge v_2} d\omega=0$ for the $v_i$'s vertical vector fields.
    \item For the proof of existence and uniqueness of $\omega_{\nabla}$ we refer to \cite{Guillemin}. Assume that $\widetilde{\omega}$ is a closed fiber-compatible and $\nabla$-compatible 2-form on $M$. Then contraction with any vertical vector annihilates the difference $\widetilde{\omega}-\omega_\nabla$ for they coincide on the fibers: $\iota_v(\widetilde{\omega}-\omega_\nabla)=0$, and since they are both closed we have $L_v (\widetilde{\omega}-\omega_\nabla)=0$. This implies that the difference of two fiber-compatible, $\nabla$-compatible closed forms must be basic, i.e. $\widetilde{\omega}-\omega_\nabla= \pi^* \omega_B$.
\end{enumerate}
\end{proof}

\begin{corollary}
Every symplectic fibre bundle admitting a closed fibre compatible form also admits a symplectic connection. In particular when the fibre bundle is strongly symplectic.
\end{corollary}

\begin{proof}
This is a direct consequence of condition $\iota_{v_1\wedge v_2}d\omega = 0$, ensured by the closedness of $\omega$.
\end{proof}

\begin{remark}
Notice that to construct a fibre-compatible form we may use partitions of unity on real smooth manifolds, but it fails in the complex setting.
\end{remark}

Let us recall now some notions from K\"ahler and Hermitian geometry.

\begin{definition}[Hermitian and K\"ahler manifolds]
Let $M$ be a complex manifold. If $M$ admits an Hermitian structure $h$ on its holomorphic tangent bundle, then $(M,h)$ will be called an \emph{Hermitian manifold}. Using $h$ we can define a complex $(1,1)$-form $\omega:=\frac{i}{2}(h- \bar{h})$ and a Riemannian metric of the underlying smooth manifold $M$, as $g= \frac{1}{2}(h + \bar{h})$. So that
$$h= g - i\omega$$
Whenever $d\omega=0$, $h$ will be called a \emph{K\"ahler structure}, the Hermitian manifold $(M,h)$ will be called a \emph{K\"ahler manifold}, $g$ and $\omega$ respectively a \emph{K\"ahler metric and form}.
\end{definition}

\begin{remark}
When $h$ is a K\"ahler structure, the induced (1,1)-form is in fact symplectic: nondegeneracy of $\omega$ is guaranteed by that of $h$. 
\end{remark}

\begin{definition}[K\"ahler fibre bundle]\label{Antonelli}
A fibre bundle $\pi: M\to B$ is called \emph{K\"ahler} if there is a locally $\partial \bar{\partial}$-exact form $\omega_{M}$ such that its restriction $\omega_{M} \vert_{F_b}=\omega_b$ to each fibre $F_b$ is a K\"ahler form. Locally, the K\"ahler form is given by 
\[\omega_{b}= \frac i 2 \sum g_{i \bar j}(z, \xi) d\, \xi^{i} \wedge d\, \bar{\xi}^{j},\]
where $(z, \xi)$ are adapted coordinates for $M$ and
\[g_{i\bar j}=\frac{\partial^2 G}{\partial \xi^i \, \bar{\xi}^j},\]
where $G$ is a (local) K\"ahler potential.
\end{definition}

The previous definition is used in \cite{Antonelli}, for instance. However, we will work with a stronger version of a K\"ahler fibration, in the spirit of Definition \ref{symfib}.

\begin{definition}[Strong K\"ahler fibre bundle]\label{Strong}
A symplectic fibre bundle $F\hookrightarrow M\twoheadrightarrow B$  is said to be \emph{strongly
K\"ahler} if the following conditions are satisfied:
\begin{enumerate}
\item $(M,\mathcal{K})$ is a K\"ahler manifold.
\item Each fibre $(F_b,\mathcal{K}_b), b\in B$ is a K\"ahler manifold.  
\item $\mathcal{K}|_{F_b}=\mathcal{K}_b$.
\item The symplectic structure associated with $\mathcal{K}_b$ coincides with the symplectic structure of the underlying symplectic fibration: $i(\mathcal{K}_b - \overline{\mathcal{K}}_b) \equiv 2\omega_b$
\end{enumerate}
\end{definition}
\begin{remark}
Definition \ref{Strong} implies Definition \ref{Antonelli}, since the local K\"ahler potential of $M$ restricts to the local trivialisation and thus induces a K\"ahler potential on fibres. For the rest of this paper, we will consider the stronger version of K\"ahler fibrations. 
\end{remark}

\begin{definition}  Let $M\to B$ be a strong K\"ahler fibre bundle. A connection $\nabla$ is called \emph{K\"ahler} if for every loop $\gamma: S^1\to B$, its parallel transport along $\nabla$ preserves the K\"ahler structure on the fibers. This means that $\mbox{Hol}_\nabla(\gamma)_b\subset \mbox{K\"ahl}(F_b)$ is a subgroup of the group of K\"ahler morphisms
\[\mbox{K\"ahl}(F_b)\cong \mbox{Symp}(F_b)\cap \mbox{Isom}(F_b).\]
\end{definition}

\begin{theorem}\label{Kahlertheorem}
Any strong K\"ahler fibre bundle admits a K\"ahler connection.
\end{theorem}

\begin{proof}
Let $M\to B$ be a strong K\"ahler fibration, with $(M,\mathcal{K})$ the global K\"ahler structure. Since the fibration is strong, the two-form $\omega$ on $M$ associated with $\mathcal K$, i.e. $2\omega = i(\mathcal{K} - \overline{\mathcal{K}})$, is symplectic. By means of Theorem \ref{SymConn}, point (1), we know that there exists an $\omega$-compatible symplectic connection $\nabla_\omega$, such that $TM=\mathrm{Ver} \oplus \mathrm{Ver}^\omega$.

The restriction to the fibre of the K\"ahler structure $\mathcal{K}\vert_{F}$ is a K\"ahler structure on the fibre, and the complex structure $J$ associated with $\mathcal{K}$ is therefore fibre-compatible. In particular, $J_F$ is an automorphism of $\mathrm{Ver}$ and $J=J_F \oplus J_{\nabla_\omega}$, with $J_{\nabla_\omega}$ an automorphism of $\mathrm{Ver}^\omega$ (cf. Lemma \ref{compatibleComplexlemma}). 

We claim that horizontal lifts of vector fields on $B$ give rise to isometries of the fibres. This condition is equivalent to 
 \begin{equation}\label{isometry}
     \left(L_{w^\#} g\right)(V_1,V_2)=0
 \end{equation}
for all $V_1,V_2\in \mathrm{Ver}$ and $w^{\#}$ the horizontal lift of $w\in\mathfrak{X}(B)$. However, $g(\cdot,\cdot) = \omega(\cdot, J\cdot)$ and $g(V_1,V_2)=\omega(V_1, J_F V_2)$ for $J$ restricts to a complex structure on the fibres. On the other hand, we know that the connection $\nabla_\omega$ is symplectic and the horizontal lift of vector fields in $B$ generates symplectomorphisms of the fibres, and Equation \eqref{isometry} is satisfied.
\end{proof}

\begin{proposition}
 Let $M\to B$ be a strong K\"ahler fibre bundle, then $\mathrm{Ver}^\omega = \mathrm{Ver}^g$ where $\mathcal{K} = g - i \omega$.
\end{proposition}

\begin{proof}
It is easy to show that, if we denote by $g$ the metric associated with $\mathcal{K}$, i.e. $2g=\mathcal{K} + \overline{\mathcal{K}}$, and $\omega$ the associated symplectic form, $2\omega= i(\mathcal{K} - \overline{\mathcal{K}})$, we have the relation $\mathrm{Ver}^g = J\left(\mathrm{Ver}^\omega\right)$. This implies 
$$TM = \mathrm{Ver}\oplus \mathrm{Ver}^\omega = J_F\left(\mathrm{Ver}\right) \oplus J\left(\mathrm{Ver}^{\omega}\right)=\mathrm{Ver} \oplus \mathrm{Ver}^g$$
and $\mathrm{Ver}^\omega=\mathrm{Ver}^g$.
\end{proof}

\begin{theorem}\label{Theo:symp+comp=kaehl}
A strong symplectic fibre bundle $(F,M,B,\omega)$ is strongly K\"ahler if and only if it admits a fibre-compatible complex structure $J$ that preserves $\omega$.
\end{theorem}

\begin{proof}
Proving that every strong K\"ahler fibre bundle is strongly symplectic is trivial, moreover the associated complex structure $J$ is compatible with $\omega=i(\mathcal{K} - \overline{\mathcal{K}})$ by construction and it is $\nabla_\omega$-compatible as argued in Theorem \ref{Kahlertheorem}.

On the other hand, let $(F,M,B,\omega)$ be a strong symplectic fibre bundle, with a complex structure $J$ such that $\omega(J\cdot,J\cdot)=\omega$. This defines a K\"ahler structure $\mathcal{K}$ on $M$ via standard methods, but to ensure that such $\mathcal{K}$ is fibre-compatible, i.e. it restricts to a K\"ahler structure on the fibres $F$, we need to require that $J$ be fibre compatible, i.e. that $J\vert_F$ be a complex structure on the fibre. 

As a matter of fact, if $J \mathrm{Ver}\subset \mathrm{Ver}^\omega$, the restriction to the fibre of the bilinear form $g:=\omega(\cdot, J\cdot)$ might fail to be nondegenerate. Then $J\mathrm{Ver} \cap \mathrm{Ver}^\omega =\{0\}$ and $J\mathrm{Ver}=\mathrm{Ver}$ is a complex structure on $F$.
\end{proof}

\subsection{Fisher structures}
In this section we propose a definition that will abstract from the example of Section \ref{Section:Orbits}. We first look at the local linear model on an even dimensional (real) vector space.

\begin{definition}(Linear Fisher structures)\label{LinFishStr}
Let $V$ be a $2n$-dimensional real vector space. A \emph{linear Fisher structure} is an automorphism $\fr{J}\in \mathrm{Aut}(V)$ such that it is skew-adjoint, i.e. $\fr{J}^t=-\fr{J}$ and whose square can be diagonalised to:  $-D=-\mathrm{diag}\{d_{i}^2\mathbb{I}_2\}$, with $1 \leq i \leq  n$. The space of such structures will be denoted by $\mathcal J_F(n)$. The eigenvalues $\{d_i\}_{i=1}^{n}$ are called the \emph{roots} of the Fisher structure.

\end{definition}

\begin{lemma}
Linear Fisher structures are normal.
\end{lemma}

\begin{proof}
This is elementary, since $\fr{J}^t = -\fr{J}$ and $\fr{J}^2 = -D$, so that $\fr{J}^t \fr{J} = \fr{J}\fr{J}^t = D$.
\end{proof}

\begin{theorem}
Let $\mathfrak J\in\mathcal{J}_F$ be a Fisher structure. Then, up to change of coordinates, the matrix representing $\mathfrak J$ in $\mathrm{GL}(2n,\mathbb{R})$ has the following form:
\[[\mathfrak J]=  \begin{bmatrix}\begin{bmatrix} 0&-d_1  \mathbb I \\d_1\mathbb I&0 \end{bmatrix}& \begin{bmatrix}0&0\\0&0 \end{bmatrix} &\cdots&\cdots & \begin{bmatrix}0&0\\0&0 \end{bmatrix}\\
 \begin{bmatrix}0&0\\0&0 \end{bmatrix}&\begin{bmatrix} 0&-d_2  \mathbb I \\d_2\mathbb I&0 \end{bmatrix}&\begin{bmatrix}0&0\\0&0 \end{bmatrix}&\cdots & \vdots\\
 \vdots&\dots&\dots&\ddots& \vdots\\
 \begin{bmatrix}0&0\\0&0 \end{bmatrix}&\cdots&\cdots&&\begin{bmatrix} 0&-d_k  \mathbb I \\d_k\mathbb I&0 \end{bmatrix}
 \end{bmatrix},\]
where $\begin{bmatrix} 0&-d_r  \mathbb I \\d_r\mathbb I&0 \end{bmatrix}$
is a $m_r\times m_r$-matrix, $m_r$ being the multiplicity of  $d_r^2$ in $D$.
\end{theorem}
\begin{proof}
It is easy to observe that the eigenvalues of $\mathfrak J$ are of the form $\pm i d_r$. Now, by choosing an eigenbasis of $\mathbb R^{2n}$, we obtain that the projection of $\mathfrak J$ to the direct sum of eigenspaces $E_{id_r} \oplus E_{-id_r}$ is $\begin{bmatrix} 0&-d_r  \mathbb I \\d_r\mathbb I&0 \end{bmatrix}$. The eigenvalues of $\mathfrak J$ are preserved by change of basis, and this concludes the proof.
\end{proof}

We say that a matrix $M$ is a complex matrix representative of $\fr{J}$ if there is a basis $\alpha$ of $\mathbb{C}^n$ for which $\fr{J}(x)=Mx$, for all $x$ in $V$.
We have the following.

\begin{proposition}\label{FishStr}
A complex matrix representative $M$ of $\fr{J}\in\mathcal J_F(n)$ in $\mathrm{GL}(n,\mathbb{C})$ if and only if the set of rows of $M$ is an orthogonal basis of $\mathbb C^n$ and $M$ is skew adjoint.
\end{proposition}

\begin{proof}
One direction follows directly from the fact that when the rows of $M$ are orthogonal, then $MM^*=D$, where $D_{ii}$ is the $\mathbb{C}$ square-norm of the $i$-th row of $M$, and therefore $M^2=-D, D>0$.
On the other hand, if $M^*=-M$ and $M^2=-D$, let $Q= D^{-1/2}$. Therefore
\begin{eqnarray*}
&(QM)(M^*Q)=\mathbb I\\
\iff& (QM)(QM)^*=\mathbb I\\
\iff& QM \in \mathcal U(n, \mathbb C)\\
 \iff& \mbox{Row}(M) \mbox{ is a positive rescaling of each vector in Row} (QM). 
\end{eqnarray*}
\end{proof}

\begin{remark}
Observe that the crucial feature of a Fisher structure is that $\fr{J}^2$ splits $\mathbb{R}^{2n}$ in (at most) $n$ 2-dimensional eigenspaces, and that its eigenvalues are negative. Notice that, since Fisher structures are normal, they are diagonalised by means of unitary matrices $U(n)=\mathrm{GL}(2n,\mathbb{R}) \cap \mathrm{Sp}(2n,\mathbb{R})$. Denoting by $d_i$ the multiplicity of the i-th eigenvalue as above, we observe that the action of $U(n)$ on $[\fr{J}]$ by conjugation has a stabiliser given by $H_{[\fr{J}]}\simeq \Pi_i U(d_k)$. Representing a Fisher structure by the diagonal matrix of eigenvalues $M=\sqrt{-D}\in \mathrm{GL}(n,\mathbb{C})$, we gather that the space $\mathcal{J}_F(n)$ can be seen as a union of $U(n)$-spaces, parametrised by purely imaginary diagonal matrices, which is the Cartan subalgebra of $\mathfrak{su}(n)$. In this sense one can characterise a Fisher structure according to which $U(n)$-space its square belongs to. This relationship between $U(n)$-spaces and Fisher structures will be explored elsewhere. 
\end{remark}

\begin{definition}[Fisher structure] \label{Fish}
Let M be a ($2n$)-dimensional manifold and $\fr{J}\colon TM \to TM$ be a linear bundle automorphism covering the identity. We will say that $\fr{J}$ is a \emph{Fisher structure on M} if $\fr{J}_p \colon T_pM\longrightarrow T_pM$ is a Fisher structure for $T_pM$, with roots $\{d_i\}_{i=1\dots n}$ for all $p\in M$, and it is a smooth function of $p$.
\end{definition}

\begin{remark}
Observe that any complex structure on $M$ is automatically a Fisher structure for the roots $d_i=1, 1\leq i \leq n$.
\end{remark}

This definition allows the following generalisation of a K\"ahler metric
\begin{definition}[Fisher-K\"ahler structure]
A Fisher-K\"ahler structure on a $2n$-dimensional manifold $M$ is a triple $(G,\Omega,\fr{J})$, with $G$ a Riemannian metric, $\Omega$ a symplectic form and $\fr{J}$ a Fisher structure such that 
$$G(\cdot, \cdot)=\Omega(\cdot,\fr{J}\cdot)$$
\end{definition}

\begin{remark}
Whereas the complex structure $J$ in a K\"ahler manifold is an isometry and preserves the symplectic structure, i.e. $G(J\cdot,J\cdot) = G$, in a Fisher K\"ahler structure this is true only up to diagonal rescalings, that is to say 
$$G(\fr{J}\cdot, \fr{J} \cdot)=-\Omega(\fr{J}\cdot, D\cdot) = G(\sqrt{D}\cdot , \sqrt{D}\cdot)$$
\end{remark}

\subsection{Examples of symplectic and K\"ahler fibrations}
\begin{ex}\label{Trivial}\emph{Trivial bundles of complex vector spaces}. Let $$\mathbb C^p\hookrightarrow \mathbb C^q\twoheadrightarrow \mathbb C^{r}$$ be the trivial vector bundle over $\mathbb C^r$ with fibre $\mathbb C^p$. The total, base and fibre spaces are all canonically symplectic with symplectic structure $dz_i\wedge d\bar{z}^i$. It is easy to check that this is a symplectic fibration and that the canonical splitting $\nabla: \alpha\mapsto (\alpha, 0)$, is symplectic, with the compatible form $\omega$ given by $\sum_{i=1}^{q} dz_i\wedge d\bar{z}^i$. This is a natural case in which the compatible connection is flat.

It is easy to observe that the K\"ahler potential 
\[G = \frac{\vert z \vert ^2} 2\]
produces the standard K\"ahler form on $\mathbb C^p, \mathbb C^q$ and $\mathbb C^r$. Thus the trivial connection is in this case a compatible K\"ahler connection and the fibre bundle is naturally K\"ahler, since the canonical complex structure $J(z)=iz$ is trivially fibre compatible. In addition to this, any linear automorphism $J_P$ which is obtained by similarity, i.e. $J_P= PJP$, with $P$ diagonal with respect to the fibre compatible basis (so that $J^2=-P^2$, where $D$ is diagonal) is a Fisher structure compatible with the fibration.
\end{ex}

\begin{ex}\label{Proj} \emph{Projective spaces}. Let $$\mathbb{CP}^p\hookrightarrow \mathbb{CP}^q\twoheadrightarrow \mathbb{CP}^{r}.$$ If $(z_i, \xi_i)$ are adapted coordinates subordinated to an open $U_i$, the local K\"ahler potential 
\[G_i= ln (\sum_{a=1}^{k+1} \vert \xi_i^a \vert^2 )\]
produces the Fubini-Study K\"ahler form on $\mathbb{CP}^p, \mathbb{CP}^q$ and $\mathbb{CP} ^r$. If $(z_i, \xi_i)$ are adapted coordinates subordinated to an open $U_i$, the local K\"ahler potential 
\[G_i= ln (\sum_{a=1}^{k+1} \vert \xi_i^a \vert^2 )\]
is fibre compatible.

\end{ex}

\begin{ex}\emph{Symplectic vector bundles}.
A natural extension of Example \ref{Trivial}  arises when we consider a vector bundle $\pi: E \to M$ such that there is a form $\Omega \in \Gamma(\bigwedge ^2 E^*)$ which restricts to a symplectic form on $E_x$, for all $x \in M.$ In other words, these represent symplectic fibrations with linear fibres. This linear version includes the tangent bundle of a symplectic manifold $TM\to M$ and the bundle $E\oplus E^*\to M$, for a given vector bundle $E\to M$. The symplectic orthogonal complement of the fibres gives an $\Omega$-compatible connection, following definition \ref{omegaCompatible}.
\end{ex}

\begin{ex} \emph{Cotangent lifts}. Given a surjective submersion 
\begin{equation}\label{cot}
\pi: X\to B
\end{equation}
with fibre $F_b$ over $b\in B$
there is a natural fibration, defined as 
\begin{equation}\label{lift}
\tilde{\pi}: M=\bigcup_{b\in B}T^*F_b \to B.
\end{equation}
The fibres are symplectic by using the canonical symplectic form for cotangent bundles and a splitting of \ref{cot} gives rise to a splitting of \ref{lift}. This indeed equips \ref{lift} with a symplectic fibration structure (see e.g. Theorem 2.1.1 in \cite{Guillemin}).
\end{ex}

\begin{remark}The examples above are particular instances of structures which are compatible with the fibre structure and an associated horizontal distribution. In particular, the Fisher structure in Example \ref{Trivial} is fibre compatible by construction. This example will be further generalised for the case of a Fisher structure on the fibre bundle of coadjoint orbits.
\end{remark}

\section{Fibre Bundle of Coadjoint Orbits}\label{Section:Fibration}
In this section we will analyse examples of fibre bundles, where the total and base spaces are co-adjoint orbits of the action of a Lie group on its dual Lie algebra.

We will approach the (symplectic) geometry of co-adjoint orbits \emph{via} the principal bundle perspective \cite{Guillemin}. Unless stated otherwise, we will let $G$ be a compact, semisimple Lie group, and $\fr{g}$ its Lie algebra.

$G$ naturally acts on the dual $\fr{g}^*$, \emph{via} the co-adjoint action, and every orbit of this action is the base space of a principal bundle, labeled by an element of the Cartan subalgebra $\rho_0\in\mathfrak t \subset \mathfrak g$:
\[H^{(\rho_{0})}\to G \to \orb{\rho_{0}},\]
where $H^{(\rho_{0})}$ is the stabiliser subgroup of $\rho_{0}$. Observe that the tangent map $ad\colon \fr{g}\longrightarrow T_{\rho_0}\orb{\rho_0}$ equivariantly identifies tangent vectors with vectors in the Lie algebra. In fact, there is an $H^{(\rho_0)}$-invariant splitting 
\begin{equation}\label{splitting}
\fr{g}=\stab{\rho_0}\oplus\nor{\rho_0}
\end{equation}
and $T_{\rho_0}\orb{\rho_0}\simeq\nor{\rho_0}$.

Given an orbit $\orb{\rho_0}$ and a point $\rho \in \orb{\rho_0}$, the bilinear form
\[\beta(K,H)_\rho:= Tr[\rho, [K,H]]\]
where $K,H\in \fr{g}$, is skew symmetric, nondegenerate and equivariant with respect to the coadjoint action. Because of the identification of tangent vectors with Lie algebra elements, this bilinear form yields a  symplectic structure on $\orb{\rho_0}$: 
$$\Omega^{(\rho_0)}_{KKS}(V,W)=\Omega_{KKS}^{(\rho_0)}\vert_{\rho}(ad_{K}\rho, ad_H\rho)=\beta(K,H)_\rho,$$ 
with $V=ad_K\rho,H=ad_H\rho$, and will be called Kirillov--Konstant--Souriau (KKS) form. 

\begin{remark}
Each orbit is promoted to the symplectic manifold $(\orb{\rho_0},\Omega_{KKS}^{(\rho_0)})$. It is a known result \cite{Borel,Reeder} that the second cohomology group of a homogeneous space as such is isomorphic to the Cartan subalgebra of the group, i.e. $H^2(G/H)\simeq \fr{t}$. The choice of an element in the Cartan subalgebra fixed an isomorphism $G/H^{(\rho_0)}\simeq \orb{\rho_0}$ and a particular class in $H^2(G/H)$.
\end{remark}

\subsection{Fibrations of orbits}
If we now choose two elements of the Cartan subalgebra $\eta_0,\rho_0\in\fr{t}$, with the condition that $\stab{\rho_0}\supset\stab{\eta_0}$, we can construct the associated bundle
\begin{equation}
    \mathcal{B}^{(\eta_0,\rho_0)}:=(G\times (Ad_{H^{(\rho_0)}}\eta_0))/H^{(\rho_0)},
\end{equation} 
One can show that 
\[\mathcal{B}^{(\eta_0,\rho_0)}\cong \orb{\eta_0}\]
is a fibre bundle over $\orb{\rho_0}$ with fibre $H^{(\rho_0)}/H^{(\eta_0)}$, i.e. we have the following sequence of homogeneous spaces:
\begin{equation}\label{fibreorbit}
    \qsp{H^{(\rho_0)}}{H^{(\eta_0)}} \longrightarrow \orb{\eta_0} \longrightarrow \orb{\rho_0}.
\end{equation}

In particular, we have the splittings $\fr{g}=\stab{\eta_0}\oplus\nor{\eta_0}$ and $\fr{g}=\stab{\rho_0}\oplus\nor{\rho_0}$. The tangent spaces (we will consider them at the reference points) are isomorphic to the normal complements $\nor{\eta_0}$ and $\nor{\rho_0}$ respectively. 

Above each orbit we can construct a \emph{tautological bundle} as follows:
\begin{definition}(Tautological bundle and tautological tangent bundle)\label{tautbundle}
Let $\eta_0$ be a Cartan element and $\orb{\eta_0}$ be its associated orbit. We define the \emph{tautological bundle} associated with the orbit $\orb{\eta_0}$ to be the vector bundle $\nablat{\eta_0}\longrightarrow \orb{\eta_0}$ with fibre above a point $\orb{\eta_0}\ni\eta=\mathrm{Ad}_g \eta_0$ given by the assignment $\nablat{\eta_0}_\eta=\mathrm{Ad}_g\nor{\eta_0}$ where $\nor{\eta_0}$ is the normal complement in $\fr{g}$ induced by $\eta_0$. By tensoring this bundle with the tangent bundle of the orbit $\orb{\eta_0}$ we define the \emph{tautological tangent bundle} $T\orb{\eta_0}\otimes\nablat{\eta_0}\longrightarrow \orb{\eta_0}$.
\end{definition}

The tautological bundle is by construction an equivariant vector bundle. Whenever an orbit is the total space of a fibration as in Equation \eqref{fibreorbit}, its tautological bundle splits.

\begin{lemma}\label{tautologicalLemma}
Let $\pi\colon\orb{\eta_0}\longrightarrow\orb{\rho_0}$ be a fibration of orbits. The tautological bundle associated to the orbit $\orb{\eta_0}$ splits as:
\begin{equation}\label{tautologicalsplitting}
\nablat{\eta_0} \simeq \fr{V}^{(\eta_0)} \oplus \nablat{\rho_0}
\end{equation}
where $\fr{V}^{(\eta_0)}\longrightarrow \orb{\eta_0}$ is a vector subbundle whose typical fibre is isomorphic to $\qsp{\stab{\rho_0}}{\stab{\eta_0}}$.
\end{lemma}

\begin{proof}
We know that the (equivariant) splitting of the Lie algebra induced by $\rho_0$, i.e. $\fr{g}=\stab{\rho_0}\oplus \nor{\rho_0}$, induces a splitting in the typical fibre of the tautological bundle $\fr{n}^{(\eta_0)}$, also equivariant:
$$\nablat{\eta_0}_{\eta}=\mathrm{Ad}_g\nor{\eta_0}=\mathrm{Ad}_g\left(\qsp{\stab{\rho_0}}{\stab{\eta_0}}\oplus \nor{\rho_0}\right)$$
with $\eta=\mathrm{Ad}_g\eta_0$. As a matter of fact, defining $\fr{V}_{\eta}^{(\eta_0)} := \mathrm{Ad}_g\qsp{\stab{\rho_0}}{\stab{\eta_0}}$ we get Equation \eqref{tautologicalsplitting}, as the splitting is equivariant, identifying $\nablat{\rho_0}_{\eta} \simeq Ad_{g}\nor{\rho_0}$.
\end{proof}

\begin{proposition}\label{Symplectic-Bott} Let $\pi\colon \orb{\eta_0} \longrightarrow \orb{\rho_0}$ be a fibration of orbits. This comes equipped with an Ehresmann connection 
\begin{equation}
    T_\eta\orb{\eta_0} =  \mathcal V_\eta \oplus \mathcal N_\eta    
\end{equation}
with $\eta\in\orb{\eta_0}$ which restricts to the Bott connection on (the annihilator of) the vertical distribution, under the identification induced by the symplectic structure: $\mathcal N_\eta\simeq\mathcal V^{\Omega_{\eta_0}}_\eta\simeq\mathcal V_\eta^0$, where $\mathcal{V}_\eta^0$ is the annihilator w.r.t. the Killing form.
\end{proposition}

\begin{proof}
By equivariance of the construction of the fibration of orbits (and of the Lie algebra splitting) we can read eq. $\fr{g}=\stab{\rho_0}\oplus\nor{\rho_0}$, together with the condition $\nor{\rho_0}\subset\nor{\eta_0}$ as the definition of an Ehresmann connection. As a matter of fact, acting with the group $G$ on the tangent space is equivalent to acting on the Lie algebra with the adjoint action, and this makes it possible to promote equation \eqref{tautologicalsplitting} to 
$$
T_\eta\orb{\eta_0} =  \mathcal V_\eta \oplus \mathcal N_\eta
$$
where $\mathcal N_\eta\simeq\mathrm{Ad}_g\nor{\rho_0}$ with $\eta=\mathrm{Ad}_g\eta_0$ and $\mathcal V_\eta\simeq\mathrm{Ad}_g\stab{\rho_0}/\stab{\eta_0}$. 

Notice that $\mathcal V_\eta$ is involutive and observe that 
\begin{equation}\label{vectlieder}
[\nor{\rho_0}, \qsp{\stab{\rho_0}}{\stab{\eta_0}}] \subset \nor{\rho_0}.
\end{equation}
Then, using the (equivariant) morphism $\mathrm{ad}_\bullet\eta\colon\nor\rho\subset\nor{\eta}\longrightarrow T_\eta\orb{\eta_0}$ we can interpret equation \eqref{vectlieder} as 
\begin{equation}
\mathcal L_{\mathcal V_\eta}\mathcal N_\eta \subset \mathcal N_\eta
\end{equation}
with $\mathcal L$ the Lie derivative, for all $\eta\in\orb{\eta_0}$. As a matter of fact, for $V\in \qsp{\stab{\rho_0}}{\stab{\eta_0}}$ and $N\in\nor{\rho_0}$, one gets the vectors $\mathrm{ad}_V\eta\in\mathcal V_\eta$ and $\mathrm{ad}_N\eta\in\mathcal{N}_\eta$ (respectively vertical and horizontal) and
\begin{equation}
\mathcal L_{\mathrm{ad}_V\eta}(\mathrm{ad}_N\eta) = [\mathrm{ad}_V\eta, \mathrm{ad}_N\eta]_{\mathfrak{X}(\orb{\eta_0})}=\mathrm{ad}_{[V,N]_{\fr{g}}}\eta
\end{equation}
but $[V,N]_{\fr{g}}\in\nor{\rho_0}$ in virtue of \eqref{vectlieder}, and $\mathrm{ad}_{[V,N]_{\fr{g}}}\eta\in \mathcal N_\eta$.

On the other hand, since for $G$ compact $\mathfrak{n}^{(\eta_0)}\simeq (\mathfrak{h}^{(\eta_0)})^0$ - with respect to the Killing form - we have that $\Omega^{(\eta_0)}_{KKS}(\mathrm{ad}_V\eta,\mathrm{ad}_N\eta)=0$ so that the symplectic orthogonal of $V_\eta$ is identified with the Horizontal subspace: 
$$\mathcal N_\eta\simeq\mathcal V^{\Omega^{(\eta_0)}_{KKS}}_\eta\simeq\mathcal V_\eta^0.$$
This, in particular, allows us to observe that the Ehresmann connection defined by the choice of Horizontal subbundle restricts to the Bott connection on the annihilator of the vertical distribution.
\end{proof}

\begin{corollary}
The tautological tangent bundle $T\orb{\eta_0}\otimes \nablat{\eta_0}$ splits as
\begin{equation}\label{tauttangentsplitting}
T\orb{\eta_0}\otimes \nablat{\eta_0}\simeq \left( \mathcal{V}\otimes \fr{V}^{(\eta_0)}\right) \oplus \left(T\orb{\rho_0}\otimes \nablat{\rho_0}\right).
\end{equation}
\end{corollary}
\begin{proof}
The proof is immediate.
\end{proof}

\begin{remark}
The tautological tangent bundle splits into the tautological tangent bundles associated to the different orbits in the fibration. As a matter of fact, the fibre of the fibre bundle for co-adjoint orbits  is also an orbit $\mathcal{O}^{(\rho_0,\eta_0)}\simeq \qsp{H^{(\rho_0)}}{H^{(\eta_0)}}$ of the stabiliser subgroup $H^{(\rho_0)}$ (cf. Equation \eqref{fibreorbit}), and $\mathcal{V}=TF=T\mathcal{O}^{(\rho_0,\trho_0)}$.
\end{remark}

Consider the equivariant map
$$\iota^{(\rho_0)}\colon \orb{\rho_0} \longrightarrow \fr{g}^*$$
Again, we have the equivariant splitting
\begin{equation}\label{splittingun}
\fr{g}=\stab{\rho_0}\oplus \nor{\rho_0}\end{equation} 
and $\nor{\rho_0}\simeq T_{\rho_0}\orb{\rho_0}$ coincides with the (pointwise) image of $\iota_*^{(\rho_0)}\colon T\orb{\rho_0} \longrightarrow \nablat{\rho_0}$, which can be regarded as a vector valued one-form $\iota_*^{(\rho_0)}\in T^*\orb{\rho_0}\otimes\mathfrak{n}^{(\rho_0)}\simeq \Omega^1(\orb{\rho_0},\nablat{\rho_0})$ and identified with the (parametrised) one-form $d\rho$, with $\rho\in\orb{\rho_0}$. 

The pre-image of 
$$\mathrm{ad}_\bullet \rho_0\colon \fr{u}(n) \longrightarrow T_{\rho_0} \orb{\rho_0},$$
is single-valued in $\nor{\rho_0}$ because of the splitting in \eqref{splittingun}. Then we have

\begin{definition}\label{DandPhimaps}
We define a bundle morphism covering the identity $\Phi^{(\rho_0)}\colon T\orb{\rho_0}\longrightarrow \nablat{\rho_0}$ via
$$\Phi_\rho^{(\rho_0)}(v)=\mathrm{Ad}_g \phi_{\rho_0}\mathrm{Ad}_{g^{-1}}(v)  $$
$\forall v \in T_{\rho}\orb{\rho_0}$ by denoting the inverse of $\mathrm{ad}_\bullet \rho_0$ as $\phi_{\rho_0}\colon T_{\rho_0} \orb{\rho_0} \stackrel{\sim}{\longrightarrow} \nor{\rho_0}$, that is 
$$\phi_{\rho_0}=\pi_{\rho_0}\circ \mathrm{ad}_\bullet\rho_0^{-1},$$
where $\pi_{\rho_0}\colon \fr{g}\longrightarrow \nor{\rho_0}$ and $\rho=\mathrm{Ad}_g \rho_0$. We can regard $\Phi^{(\rho_0)}$ as a vector valued differential form $\Phi^{(\rho_0)}\in  \Omega^1(\orb{\rho_0},\nablat{\rho_0})$. 

We further define $\fr{D}^{(\rho_0)}\colon \nablat{\rho_0} \longrightarrow\nablat{\rho_0}$ to be the bundle morphism that coincides on each fibre with the adjoint action seen as an endomorphism: $\mathrm{ad}_\bullet \rho_0\colon \nor{\rho_0} \longrightarrow \nor{\rho_0}$, i.e. if $\rho=\mathrm{Ad}_g \rho_0$
\begin{equation}\label{Dfrak}
    \fr{D}^{(\rho_0)}_{\rho} = \mathrm{ad}_{\bullet} \rho.
\end{equation}
\end{definition}

\begin{lemma}
We have that
\begin{equation}
    \iota_*^{(\rho_0)} \equiv \fr{D}^{(\rho_0)} \circ \Phi^{(\rho_0)}
\end{equation}
\end{lemma}
\begin{proof}
The proof is immediate.
\end{proof}

\begin{lemma}\label{Phicompatible}
Let $\pi\colon\orb{\eta_0}\longrightarrow\orb{\rho_0}$ be a fibration of co-adjoint orbits. Then, the map $\Phi^{(\eta_0)}\colon T\orb{\eta_0} \longrightarrow \nablat{\eta_0}$ is compatible with the corresponding splittings of $T\orb{\eta_0}$ and $\nablat{\eta_0}$.
\end{lemma}
\begin{proof}
The section $\Phi^{(\eta_0)}\in\Omega^1(\orb{\eta_0},\fr{n}^{\eta_0})$ is $\nabla$-compatible by construction, in fact on each fibre it assigns to each tangent vector its preimage along the adjoint action:
$$\Phi^{(\eta_0)}_{\eta}\equiv \mathrm{Ad}_g\phi_{\eta_0}\mathrm{Ad}_{g^{-1}}\colon T_{\eta}\orb{\eta_0}\longmapsto \nor{\eta}$$
Since the inclusion $\nor{\rho_0}\subset\nor{\eta_0}$ yields the $\stab{\rho_0}$-equivariant splitting $\nor{\eta_0}=\nor{\rho_0} \oplus \stab{\rho_0}/\stab{\eta_0}$ we can conclude that the map $\Phi^{(\eta_0)}$ splits accordingly.
\end{proof}

\subsection{Sympletic and K\"ahler characterisation of orbits}

Now that we have introduced all of the relevant constructions for our purposes we can start characterising the fibration of co-adoint orbits. Let us recall that $G$ is a compact Lie group throughout.

\begin{theorem}\label{theosymp}
The fibre bundle of coadjoint orbits $\orb{\eta_0}\to \orb{\rho_0}$ is strongly symplectic.
\end{theorem}
\begin{proof}
The proof is based on Theorem 2.3.3 in \cite{Guillemin}, and it goes by showing that the vertical subspace $\mathcal{V}_\eta\simeq T_\eta \left(\mathrm{Ad}_{H^{(\rho_0)}}\eta_0\right)$ is a symplectic subspace of $T_\eta\orb{\eta_0}$. The point $\rho_{0}$ of $\mathfrak g$  with stabiliser subgroup $H^{(\rho_0)}$ induces the $H^{(\rho_0)}$-equivariant inclusion 
\[\iota_F\colon \stab{\rho_0} \to \fr{g}\]
and the $H^{(\rho_0)}$-equivariant isomorphism
$$\iota^*\colon \mathrm{Ad}_{H^{(\rho_0)}}\eta_0 \stackrel{\sim}{\longrightarrow} \mathrm{Ad}_{H^{(\rho_0)}} \iota^*(\eta_0)$$
Observe that the right-hand side of this morphism is a co-adjoint orbit of the stabiliser subgroup $H^{(\rho_0)}$, and it is therefore a symplectic manifold. Then, for every $X,Y\in\stab{\rho_0}$ we get that 
\begin{equation}
\Omega_{KKS}^{(\eta_0)}(\mathrm{ad}_X \eta_0,\mathrm{ad}_Y \eta_0)= \Omega_{KKS}^{(\iota^*\eta_0)}(\mathrm{ad}_X\iota_*\eta_0, \mathrm{ad}_Y\iota_*\eta_0),\end{equation}
thus $T_\eta \left(\mathrm{Ad}_{H^{(\rho_0)}}\eta_0\right)$ is a symplectic subspace of $T_\eta \orb{\eta_0}$, meaning that the fibres are symplectic. Proposition \ref{Symplectic-Bott} shows that the horizontal and vertical distributions are symplectically orthogonal, thus proving that there is an $\Omega_{KKS}$-compatible connection, which is symplectic by means of Theorem \ref{SymConn} and $\orb{\eta_0} \longrightarrow \orb{\rho_0}$ is a strong symplectic fibre bundle.
\end{proof}

In what follows we will show how it is possible to endow the symplectic fibration of co-adjoint orbits of a compatible complex (and K\"ahler) structure. It is well known \cite{BFR, Kronheimer, Neeb} that each co-adjoint orbit is independently a K\"ahler manifold. Here we present the complex structure on orbits via an alternative construction which is adapted to the symplectic fibration we have outlined so far.
\begin{theorem} \label{TH:compatiblecomplexstructure}

Every symplectic fibre bundle of coadjoint orbits $\pi\colon \orb{\eta_0} \longrightarrow \orb{\rho_0}$ admits a $\nabla$-compatible complex structure $J^{(\eta_0)}\colon T\orb{\eta_0}\longrightarrow T\orb{\eta_0}$. 
\end{theorem}

\begin{proof}
Theorem \ref{theosymp} shows that the choice of two points in the Cartan subalgebra $\eta_0$ and $\rho_0$, such that $\mathfrak{h}_{{\eta_0}}\subset\mathfrak{h}_{\rho_0}$, induces a symplectic fibration with connection. The symplectic connection is once again given by the splitting $\mathfrak{g}=\mathfrak{h}_{{\rho_0}}\oplus\mathfrak{n}_{{\rho_0}}$. Observe that $\mathfrak{n}_{\rho_0}$ coincides with the elements in $\mathfrak{g}$ that do not vanish on $\rho_0$ (under the adjoint action).

Assume that $\eta_0$ is a generic point in the Cartan subalgebra. Since the group is compact, the stabiliser of such generic point (a maximal torus) is contained in the stabiliser of any other non-generic point $\rho_0$, meaning that $\orb{\eta_0}$ is always the total space of a fibration over $\orb{\rho_0}$. 

The stabiliser subalgebra of a generic point coincides with the (real) Cartan subalgebra itself, meaning that the associated splitting for the complexified Lie algebra $\mathfrak g^{\mathbb C}$ reads:
\begin{equation}
\mathfrak{g}^{\mathbb C}=\mathfrak{h}^{\mathbb C}_{\eta_0}\oplus\mathfrak{n}_{\eta_0}^{\mathbb C}\simeq \mathfrak{h}_{\eta_0}^{\mathbb C}\oplus \bigoplus_{\alpha}\mathfrak{g}_\alpha^{\mathbb C}
\end{equation}
where $\mathfrak{n}^{\mathbb C}_{\eta_0}\simeq\bigoplus_{\alpha}\mathfrak{g}^{\mathbb C}_\alpha$, and $\mathfrak{g}_\alpha^{\mathbb{C}}$ denotes the root space associated with root $\alpha$. Let $\{V_\alpha\}$ be a basis of $\mathfrak{g}_\alpha^{\mathbb C}$, and define
\begin{equation}
    X_\alpha = i (V_{\alpha} + V_{-\alpha}); \qquad Y_\alpha = V_{\alpha} - V_{-\alpha},
\end{equation}
with $\alpha\in\Delta$ the space of simple positive roots, so that $\{X_\alpha,Y_{\alpha}\}_{\alpha\in\Delta}$ is a (real) basis of $\mathfrak{n}_{\eta_0}$. Consider the map 
\begin{equation}
\widetilde{J}\colon \bigoplus_{\alpha}\mathfrak{g}_\alpha^{\mathbb{C}} \longrightarrow \bigoplus_{\alpha}\mathfrak{g}_\alpha^{\mathbb{C}};\ \ \ V_{\pm\alpha} \longmapsto \pm i V_{\pm\alpha}.
\end{equation}
On the basis $\{X_{\alpha},Y_{\alpha}\}$ we have $\widetilde{J}(X_{\alpha})=-Y_{\alpha}$ and $\widetilde{J}(Y_{\alpha})=X_{\alpha}$, 
which then defines an equivariant complex structure on $\orb{\eta_0}$, since $\mathfrak{n}_{\eta_0}\simeq T_\rho\orb{\eta_0}$, and the splitting is equivariant. We will denote it equivalently by $\widetilde{J}\equiv J^{(\eta_0)}$ for some generic point $\eta_0$ in the Cartan subalgebra, a notation that will become useful clear in a moment.
Given a root $\alpha$,  $J^{(\eta_0)}$ 
leaves $\mathfrak g_{\alpha}^{\mathbb C}\oplus \mathfrak g_{-\alpha}^{\mathbb C}$ invariant by construction, as it acts by multiplication by $\pm i$ on $g_{\pm\alpha}^\mathbb{C}$.

Now, since $\mathfrak{n}_{\rho_0}\subset\mathfrak{n}_{\eta_0}$ we can induce a complex structure on $\orb{\rho_0}$, compatible with the fibration $\orb{\eta_0}\longrightarrow\mathcal{O}_{{\rho}_0}$. There will be root spaces $\mathfrak{g}^{\mathbb{C}}_\beta$ such that $[\mathfrak{g}_\beta^{\mathbb{C}},\rho_0]=0$, however, if $[\mathfrak{g}_\beta^{\mathbb{C}},\rho_0]\equiv \beta(\rho_0)=0$, then also $\mathfrak{g}_{-\beta}^{\mathbb{C}}$ will have the same property. The map $J^{(\eta_0)}$ then restricts to 
$$J^{(\rho_0)}\colon\bigoplus_{\alpha, \alpha(\rho_0)\not=0}\mathfrak{g}^{\mathbb C}_\alpha\longrightarrow \bigoplus_{\alpha, \alpha(\rho_0)\not=0}\mathfrak{g}^{\mathbb C}_\alpha,$$
but $\bigoplus_{\alpha, \alpha(\rho_0)\not=0}\mathfrak{g}_\alpha^{\mathbb{C}}\simeq\mathfrak{n}_{\rho_0}^{\mathbb{C}}\simeq \mathcal{N}_{\rho_0}^{\mathbb{C}}$, and $J^{(\rho_0)}$ then defines an equivariant complex structure on $\orb{\rho_0}$ via the identification of $T\orb{\rho_0}$ with the horizontal distribution $\mathcal{N}$ induced by the symplectic connection. On the other hand, the map $\widetilde{J}$ also restricts to $\bigoplus_{\alpha, \alpha(\rho_0)=0}\mathfrak{g}_\alpha\simeq\qsp{\mathfrak{h}_{\rho_0}}{\mathfrak{h}_{\eta_0}}\simeq \mathcal{V}_{\eta_0}$ and therefore we conclude that 
$$J^{(\eta_0)} = J^{(\eta_0)}\vert_{\mathcal{N}}\oplus J^{(\eta_0)}\vert_{\mathcal{V}} = J^{(\rho_0)}\oplus J^{(\eta_0)}\vert_{\mathcal{V}},$$ 
which means that $J^{(\eta_0)}$ is $\nabla$-compatible.

Extending the argument, we can easily gather that an analogous compatibility holds also when two non-generic points $\rho_0,\widetilde{\rho}_0\in\fr{t}^+$ are chosen to generate the fibration, as long as $\stab{\rho_0}\subset\stab{\widetilde{\rho}_0}$ or vice-versa. 
\end{proof}

\begin{theorem}\label{TH:kaehlerfibration}
The fibre bundle of coadjoint orbits $\orb{\eta_0}\to \orb{\rho_0}$ is strongly K\"ahler.
\end{theorem}

\begin{proof}
The fibration $\orb{\eta_0}\to \orb{\rho_0}$ is strongly symplectic by means of Theorem \ref{theosymp} and therefore admits a symplectic connection $\nabla$. In virtue of Theorem \ref{TH:compatiblecomplexstructure} it also admits a $\nabla$-compatible complex structure, and since the base of the fibration is a symplectic manifold (hence even-dimensional), such complex structure is also fibre-compatible by means of Lemma \ref{compatibleComplexlemma}. The results then follows from Theorem \ref{Theo:symp+comp=kaehl}, as any strongly symplectic fibration with a fibre-compatible complex structure is strongly K\"ahler.
\end{proof}

\begin{proof}[Alternative proof]
Consider a fibration of orbits as above: $\pi\colon\orb{\eta_0} \longrightarrow \orb{\rho_0}$. We have shown how to construct symplectic and complex structures $\Omega^{(\eta_0)},J^{(\eta_0)}$ compatible with the fibration.

Now, on $\orb{\eta_0}$, let $G^{(\eta_0)}$ be the equivariant metric $G^{(\eta_0)}={\Omega^{(\eta_0)}}(\cdot,{J}^{(\eta_0)}\cdot)$, and similarly for any other lower dimensional orbit $\orb{\rho_0}$. 

The splitting of the tangent space $T\orb{\eta_0}$ into Horizontal and Vertical components given by the symplectic connection $T\orb{\eta_0}=\mathcal V\oplus \mathcal N$, $\mathcal{N}_{\eta}\simeq \nor{\eta_0}$ is an orthogonal splitting with respect to $G^{(\eta_0)}$, since $J^{(\eta_0)}\colon \mathcal{N} \longrightarrow \mathcal{N}$.

By construction, then, the K\"ahler structure given by $({G}^{(\eta_0)},{\Omega}^{(\eta_0)},{J}^{(\eta_0)})$ is compatible with the fibration, i.e. it restricts to a K\"ahler structure on the fibres and its restriction to the Horizontal bundle is also related to a K\"ahler structure on the base space via horizontal lift. 
\end{proof}

\section{Coadjoint orbits and quantum information}\label{Section:Orbits}
In this section we will focus on an important particular example of the construction outlined in Section \ref{Section:Fibration}. We will consider the group $G$ to be the n-dimensional unitary group $U(n)$ and we will look at its action on its (dual) Lie algebra $\fr{u}(n)^* \simeq i\fr{u}(n)$, identified with Hermitian matrices. In particular, we will consider the (stratified) space of positive-definite trace-one Hermitian matrices, meaning that not all orbits will be suitable for the construction of this chapter. 

In this context we will be able to rephrase ordinary finite-dimensional quantum mechanics and translate several objects that are relevant for quantum information theory, such as the quantum Fisher information index (here we interpret a multi-parameter generalisation of it as a symmetric tensor on co-adjoint orbits). A thorough account on this can be found in \cite{CES}, while for more on the geometry of the space of mixed states as a stratified space, we refer to \cite{GKM}.

\subsection{Orbit description of mixed states}
In quantum mechanics, the concept of a mixed state, i.e. a superposition of \emph{pure} states, is mathematically formulated as a positive, trace-1 Hermitian matrix $\rho$, which can be seen as a convex combination of rank-1 projectors. Its eigenvalues are the probabilities that a given measurement will yield the value assigned to the respective element of a basis of the underlying Hilbert space (thus they should sum to $1$). Pure states can also be seen as (maximally dengenerate) mixed states, when the probability of finding a state along a given basis element is $1$, that is when $\rho$ is itself a projector.

If we identify the space of $n$-dimensional Hermitian matrices with the dual space to the Lie algebra $\fr{u}(n)$, using the pairing $(X,Y)=i\mathrm{Tr}[XY]$, we can immediately see that such space bears a representation of $U(n)$ that for all practical purposes will be identified with the co-adjoint action.

A mixed state is then an element $\rho\in i\fr{u}(n)$, positive definite and with trace equal to $1$. We denote the space of $n$-dimensional mixed states by $\mathcal{D}(n)\subset \fr{u}(n)^*$.

The very fact that Hermitian matrices are diagonalisable is an embodiment of the fact that each and every co-ajoint orbit intersects the Cartan subalgebra in a finite number of points (exactly the cardinality of the Weyl group, in this case $S_n$). In fact, the Cartan subalgebra of $\fr{u}(n)$ is the algebra of diagonal (anti-)Hermitian matrices $\fr{t}_n$ and the intersection points are the matrices of (permuted) eigenvalues.

Choosing an ordering of the eigenvalues is equivalent to specifying a (positive) Weyl-chamber $\fr{t}^+\subset\fr{t}_n$, and that makes the intersection unique. The diagonal representative $\rho_0=\mathrm{diag}\{\lambda_1\dots \lambda_n\}$ is then an element of a positive Weyl chamber and the reference point for the associated co-adjoint orbit $\orb{\rho_0}$. 

Denote by $\stab{\rho_0}$ the subalgebra of Hermitian matrices that commute with $\rho_0$. The Lie bracket is given by $[A,B]=i(AB - BA)$.

\begin{definition}
Let $\fr{t}^+$ be a positive Weyl chamber in the Cartan subalgebra $\fr{t}_n$ of $i\fr{u}(n)$. A $\Lambda$-mixed state is a positive definite element $\rho_0\in\fr{t}^+$, such that $\mathrm{Tr}\rho_0=1$, $\rho_0=\mathrm{diag}(\Lambda)$, $\Lambda=\{\lambda_i\}_{i=1\dots n}$. The associated co-adjoint orbit $\orb{\rho_0}$ will be called the space of $\Lambda$-mixed states.
\end{definition}

\begin{remark}
Observe that one does not expect the set of positive, trace-$1$ diagonal Hermitian matrices to bear a subalgebra stucture in $i\fr{u}(n)$. Nevertheless, every point in that set - i.e. a mixed state - will belong to the Cartan subalgebra of $i\fr{u}(n)$, and will therefore single-out an $U(n)$ co-adjoint orbit.
\end{remark}

\subsection{The Fisher tensor on a single orbit.}

In this section we will present the construction of the Fisher information tensor on orbits of the unitary group.

Consider the root space decomposition $\fr{u}(n)^{\mathbb C}=\fr{t}^{\mathbb C}_n\bigoplus_{\alpha\in\Delta}\fr{u}_\alpha^{\mathbb C}$, where $\fr{u}^{\mathbb C}_\alpha$ denotes the root space associated with the root $\alpha$. A root vector is an element $X^\alpha\in\fr{u}_\alpha$ satisfying the basic equation $\mathrm{ad}_{X^\alpha} H = \alpha(H)X^\alpha$ for all $H\in\fr{t}_n$. 

The sum of the root spaces is parametrised by off-diagonal Hermitian matrices, with basis $\{v_{ij}\}$ of matrices with a $1$ in the $(i,j)$ entry ($i\not=j$) and zero elsewhere. If we denote by $\mathrm{e}_i$ functionals on diagonal matrices such that  $\mathrm{e}_i(\rho_0)=[\rho_0]_{ii}=\lambda_i$, then a basis of the root system for $\fr{u}(n)^{\mathbb C}$ is given by $\{\alpha_{ij}=\mathrm{e}_i-\mathrm{e}_j\}$.

Denote by $\Delta_{\rho_0}$ the set of roots that do not vanish on $\rho_0$, i.e. $\alpha(\rho_0)\not=0$ and, dually, by $R_{\rho_0}$ the set of root vectors whose adjoint action on $\rho_0$ does not vanish, i.e. $\mathrm{ad}_{X^\alpha}\rho_0=[X^\alpha,\rho_0]\not=0$.  We have that $\alpha\in\Delta_{\rho_0} \iff X^\alpha\in R_{\rho_0}$. The span of all vectors in $R_{\rho_0}$ generates $\nor{\rho_0}$ so that $\alpha_{ij}\not\in\Delta_{\rho_0} \iff v_{ij}\not\in\nor{\rho_0}$. 

Recalling Definition \ref{DandPhimaps} for the map $\fr{D}^{(\rho_0)}$, its action on the basis $\{v_{ij}\}$ reads
$$\fr{D}^{(\rho_0)}_{\rho_0}\equiv \mathrm{ad}_{\bullet}\rho_0:v_{ij}\mapsto \alpha_{ij}(\rho_0) v_{ij} = (\lambda_i-\lambda_j)v_{ij}$$

The reason for this rewriting of the adjoint action will be clear considering the following
\begin{lemma}[\cite{CES,ES2}]
Let $\rho_0$ be a $\Lambda$-mixed state and $A\in\nor{\rho_0}$. Equation
\begin{equation}\label{LinearLequation}
    A=\frac12\{X,\rho_0\},
\end{equation}
where $\{A,B\}=AB + BA$, has a unique solution $X\equiv L_{\rho_0}(A)\in\nor{\rho_0}$ for all $A\in\nor{\rho_0}$. This defines an endomorphism $L_{\rho_0}\colon\nor{\rho_0} \longrightarrow \nor{\rho_0}$ and consequently a bundle morphism $\fr{L}^{(\rho_0)}\colon \nablat{\rho_0}\longrightarrow \nablat{\rho_0}$. In a basis of off-diagonal Hermitian matrices the map $L_{\rho_0}$ reads
$$L_{\rho_0}\colon v_{ij}\mapsto \frac{2}{(\lambda_i+\lambda_j)} v_{ij}$$
\end{lemma}
\begin{remark}
Observe that the condition $\lambda_i + \lambda_j =0$ implies $\lambda_i=\lambda_j=0$ and we have that the domain of $L$ is restricted, i.e. $v_{ij}\not\in\nor{\rho_0}$. Moreover, the pairwise sums of eigenvalues of a diagonal matrix, given by the functionals $\beta_{ij}=e_i + e_j$ can be expressed in terms of the roots of the Lie algebra and the trace operator:
\begin{equation}\label{Map}
    \beta_{ij}=\frac{1}{n}\left(\sum_{(km)\not=(ij)} \alpha_{(km)} + 2\left(n - \left\lfloor\frac{n+2}{2}\right\rfloor \right)\mathbb{T}\right)
\end{equation}
where $\alpha_{ij}(\rho_0)=\lambda_i - \lambda_j$, the sum is on the ordered pairs indexed by $(km)$, and  $\mathbb{T}(\rho_0)=\mathrm{Tr}(\rho_0)$. Interestingly, the functionals $\beta_{ij}$ are analogous to roots under the replacement of matrix commutation with anticommutation:
$$\{v_{ij},H\} = \beta_{ij}(H)v_{ij}.$$
\end{remark}

\begin{remark}\label{Remark:notall}
The construction we outlined here is well defined for all choices of $\rho_0$ only because of the restrictions $\rho_0 > 0$ and $\mathrm{Tr}[\rho_0]=1$. In general, the well-definiteness of the map $L_{\rho_0}$ is not automatically ensured for all points in the Weyl chamber. In fact, it is sufficient to consider the orbit of the point $\xi_0=\mathrm{diag}\{k,-k\}$ in the Cartan subalgebra of $\fr{su}(2)$ to see that the map $L_{\xi_0}\colon \nor{\xi_0}\longrightarrow\nor{\xi_0}$, as the solution of Equation \eqref{LinearLequation}, is not defined.
\end{remark}

The maps $\fr{D}^{(\rho_0)}$ and $\fr{L}^{(\rho_0)}$ are endomorphisms of the tautological bundle $\nablat{\rho_0}$ that cover the identity, however we can consider the following construction.

\begin{definition}($\mathbb{D}$ and $\mathbb{L}$-maps \cite{CES})\label{DLmaps}  We define the maps $$\mathbb{D}^{(\rho_0)},\mathbb{L}^{(\rho_0)}\colon T\orb{\rho_0}\otimes \nablat{\rho_0} \longrightarrow T\orb{\rho_0}\otimes \nablat{\rho_0}$$ 
via $\mathbb{D}=\mathrm{id}\otimes \fr{D}^{(\rho_0)}$ and $\mathbb{L}=\mathrm{id}\otimes \fr{L}^{(\rho_0)}$.
\end{definition}

\begin{remark}
By definition of the $\fr{D}^{(\rho_0)}$ map, we have that $\iota_*\vert_{\rho} (V_\rho) = D_\rho \circ \phi (V_\rho)$, so that, denoting $\iota_*\vert_{\rho}\equiv d\rho$,  we have $d\rho = (\mathbb{D}^* \Phi)_\rho$.
\end{remark}

\begin{propdef}[\cite{CES,ES2}]
Let $\rho_0$ be a $\Lambda$-mixed state, and consider the following equation in $\Omega^1(\orb{\rho_0},\fr{n}^{\rho_0})^{U(n)}$.
\begin{equation}
    d\rho = \frac12 \{d_\ell\rho, \rho\}
\end{equation}
There exists one and only one solution $d_\ell\rho\in\Omega^1(\orb{\rho_0},\fr{n}^{\rho_0})^{U(n)}$, called the \emph{Symmetric Logarithmic differential} of $\rho_0$. Moreover, $d_\ell\rho=\mathbb{L}^*d\rho\equiv(\mathbb{D}\circ\mathbb{L})^*\Phi$.
\end{propdef}

In this section we will use the construction outlined so far to define the notion of Fisher information tensor on the spaces of $\Lambda$-mixed states.

\begin{definition}[Fisher information tensor, form and metric \cite{CES}]
Let $\rho_0$ be a $\Lambda$-mixed state. We define the \emph{Fisher information tensor} $\fr{F}^{(\rho_0)}$ for the mixed state $\rho_0$ to be the section of $(T^*\orb{\rho_0})^{\otimes 2}$ given by
\begin{equation}
    \fr{F}^{(\rho_0)}\vert_{\rho}=\mathrm{Tr}\left[\rho d_\ell\rho\otimes d_\ell\rho\right]
\end{equation}
Denote by $\fr{W}^{(\rho_0)}$ and $\fr{G}^{(\rho_0)}$ respectively the antisymmetric and symmetric part of the Fisher tensor. We will call them respectively Fisher form and Fisher metric.
\end{definition}

\begin{theorem}[\cite{CES}]
The Fisher form is a symplectic form, in particular
$$\fr{W}^{(\rho_0)}=\left(\mathbb{DL}\right)^* \Omega^{(\rho_0)}_{KKS}.$$
\end{theorem}

We conclude the section with the following
\begin{definition}
The Fisher form induces the Poisson bivector $\fr{P}^{(\rho_0)}\in \bigwedge^2(T\orb{\rho_0})$ via the inverse $\fr{P}^{(\rho_0)}=(\fr{W}^{(\rho_0)})^{-1}$. We can then construct the tensor \begin{equation}
    \fr{J}^{(\rho_0)}:= \mathsf{C}_{(1,1)}(\fr{F}^{(\rho_0)},\fr{P}^{(\rho_0)})\in \mathbb{T}^{1,1}\orb{\rho_0}
\end{equation} 
where $\mathsf{C}_{(1,1)}$ is the tensor-contraction operator that outputs a section of $\mathbb{T}^{1,1}\orb{\rho_0}$.
\end{definition}

\begin{theorem}
Let $\rho_0$ be a $\Lambda$-mixed. The Fisher information tensor on $\orb{\rho_0}$ is a Fisher-K\"ahler structure for all $\rho_0\in\fr{t}$ and the $(1,1)$-tensor $\fr{J}^{(\rho_0)}$ is a Fisher structure. It reduces to a complex structure if and only if $\Lambda=\{1,0\dots,0\}$, in which case the Fisher information tensor is a K\"ahler structure.
\end{theorem}

\begin{proof}
Pick $\rho_0=\mathrm{diag}\{\lambda_1,\dots,\lambda_n\}$. Let $V_A$ be an open in $\orb{\rho_0}$ and pick the local chart $\varphi_A$ given by exponentiation of the Lie algebra on $V_A$. If $\{\sigma_I\}$ is a basis of off-diagonal Hermitian matrices in $\nor{\rho_0}$ (the index $I$ denoting an ordered pair of indices $I=(i,j)$ with $1\leq i<j\leq n$) we can write 
$$U(z^{(A)})=\exp\{i\sum_{I}z^{(A)}_I \sigma_I\}$$
with $z^{(A)}_I$ complex numbers in the chart $(V_A,\varphi_A)$ Then we have that
$$
d\rho_0 = \left(\begin{array}{cccc}
    0 & \dots & \alpha_I(\rho_0) [dz^{(A)}_I]^* & \dots \\
    \vdots & \ddots & \vdots & \vdots \\
    \alpha_I(\rho_0) dz^{(A)}_I & \dots & 0 & \vdots \\
    \vdots & \dots & \dots & 0
\end{array}\right)
$$
If follows from a simple computation that (we drop the superscript $(A)$)
\begin{equation}
\fr{F}_{\rho_0}=\sum_I \left(2\frac{\alpha_I}{\beta_I} \right)^2\left(\beta_I dz_I \odot dz_I^* - i\alpha_I dz_I \wedge dz_I^* \right)
\end{equation}

It is easy to read the symmetric and anti symmetric parts of the Fisher tensor from the previous formula. Consequently, we can construct the bivector $\fr{P}^{(\rho_0)}$ as
\begin{equation}
    \fr{P}^{(\rho_0)}= i\sum_I \left(2\frac{\alpha_I}{\beta_I} \right)^{-2}\alpha_I^{-1} \frac{\partial}{\partial z_I}\wedge \frac{\partial}{\partial z_I^*}
\end{equation}
so that 
\begin{equation}\label{generalcomplex}
    \fr{J}^{(\rho_0)}=i\frac{\beta_I}{\alpha_I} \left(dz_I\otimes \frac{\partial}{\partial z_I} + dz_I^*\otimes \frac{\partial}{\partial z_I^*}\right)
\end{equation}

Looking at Equation \eqref{generalcomplex} it is easy to see that it squares to 
$$(\fr{J}^{(\rho_0)})^2 = - {\Delta}^{(\rho_0)}$$
where the tensor $\Delta^{(\rho_0)}$ acts on a chart by a multiple of the identity $\Delta^{(\rho_0)}_I= \left(\frac{\beta_I}{\alpha^I}\right)^2 \mathbb{I}$. Thus, we can conclude that $\fr{J}^{(\rho_0)}$ is a complex structure if and only if $\frac{\beta_I}{\alpha_I}=1$ for all indices $I$, a condition that is satisfied only by $\rho_0=\{1,0,\dots,0\}$.

More generally, it follows from the construction that $\fr{G}^{(\rho_0)}(\cdot,\cdot)=\fr{W}^{(\rho_0)}(\cdot,\fr{J}^{(\rho_0)}\cdot)$ and the Fisher information tensor yields the Fisher-K\"ahler structure $(\fr{G}^{(\rho_0)},\fr{W}^{(\rho_0)},\fr{J}^{(\rho_0)})$ for all $\Lambda$-mixed states $\rho_0$. In particular, it yields a K\"ahler structure for a degenerate $\rho_0$.

\end{proof}

\subsection{The Fisher tensor on the fibration of orbits}
In this section we apply the general theory of fibration of co-adjoint orbits to the case of mixed states, specifying to the group $U(n)$ and adapting the discussion of Section \ref{Section:Fibration}. This will allow us to promote the constructions of the previous section to the fibre bundle of co-adjoint orbits corresponding to mixed states.

Anytime we are given two $\Lambda$-mixed states,say $\rho_0$ and $\trho_0$ with the property that $\stab{\trho_0}\subset\stab{\rho_0}$ we know we can construct the fibre bundle
$$\pi\colon\orb{\trho_0}\longrightarrow \orb{\rho_0}$$
and due to Theorem \ref{TH:kaehlerfibration} we know that it is endowed with a compatible K\"ahler structure and a symplectic connection, coming from the equivariant splitting $i\fr{u}(n)=\stab{\rho_0}\oplus \nor{\rho_0}$. 

\begin{remark}
Observe that in a fibration of $\Lambda$-mixed states, all the spaces involved are $\Lambda$-mixed spaces of states, for different states. For instance, when we are considering $\trho_0=\mathrm{diag}\{\lambda_1,\dots,\lambda_n\}$ for $\lambda_i\not=\lambda_j$ and $\rho_0=\mathrm{diag}\{1,0,\dots,0\}$, the fibres of $\orb{\trho_0}\to\orb{\rho_0}$ are given by the orbits of $\trho_0$ under the action of the subgroup $H^{(\rho_0)}$.
\end{remark}

Recalling the bundle morphisms of Definition \ref{DLmaps} we have the following.
\begin{proposition}\label{DLcompatible}
The maps $\mathbb{D}^{\trho_0}$ and $\mathbb{L}^{\trho_0}$ are $\nabla$-compatible, i.e. 
$$\mathbb{D}^{\trho_0} = \mathbb{D}^{\trho_0}_{\fr{V}} \oplus \mathbb{D}^{\rho_0}_{\fr{N}}$$
where $\mathbb{D}^{\trho_0}_{\fr{V}}\colon \mathcal{V}\otimes \fr{V}^{\trho_0}\longrightarrow \mathcal{V}\otimes \fr{V}^{\trho_0}$, while $\mathbb{D}^{\rho_0}_{\fr{N}}\colon \mathcal{N}\otimes \fr{N}^{\rho_0}\longrightarrow \mathcal{N}\otimes \fr{N}^{\rho_0}$, and similarly for $\mathbb{L}^{\trho}$.
\end{proposition}

\begin{proof}
The $\nablat{\trho_0}$-bundle automorphisms $\fr{D}^{(\trho_0)}$ and $\fr{L}^{(\trho_0)}$ are compatible with the splitting of Equation \eqref{tautologicalsplitting}. In fact, the maps are equivariant and we can specialise to the reference point $\trho_0$, i.e. we can look at 
$$\fr{D}^{(\trho_0)}\vert_{\trho_0}, \fr{L}^{(\trho_0)}\vert_{\trho_0}\colon \nor{\trho_0}\longrightarrow\nor{\trho_0}$$ 
with, e.g. $\fr{D}^{(\trho_0)}\vert_{\trho_0}\equiv \mathrm{ad}_\bullet \trho_0$. However, we know that such maps just act by rescaling the basis $\{v_{ij}\}$ of $\nor{\trho_0}$, which is compatible with the splitting $\nor{\trho_0}=\qsp{\stab{\rho_0}}{\stab{\trho_0}}\oplus \nor{\rho_0}$. In fact, we have that $v_{ij}\in\nor{\trho_0}\iff \mathrm{ad}_{v_{ij}}\trho_0\not=0$, then
\begin{equation}
    \begin{cases}
    \mathrm{ad}_{v_{ij}}\rho_0=0 \iff v_{ij}\in \qsp{\stab{\rho_0}}{\stab{\trho_0}}\\
    \mathrm{ad}_{v_{ij}}\rho_0\not=0 \iff v_{ij}\in \nor{\rho_0}
    \end{cases}
\end{equation}
Looking at Lemma \ref{tautologicalLemma} this means that $\fr{D}^{(\trho_0)}=\fr{D}^{(\trho_0)}\vert_{\fr{V}} \oplus \fr{D}^{(\trho_0)}\vert_{\nablat{\rho_0}}$ and consequently, since $\mathbb{D}^{(\trho_0)}=\mathrm{id}\otimes \fr{D}^{(\trho_0)}$, and similarly for $\fr{L}^{(\trho_0)}$ and $\mathbb{L}^{(\trho_0)}$, it proves the claim.
\end{proof}

We can understand the behaviour of the Fisher tensor with respect to a fibration of orbits with the following
\begin{proposition}
Let $\pi\colon \orb{\trho_0}\longrightarrow \orb{\rho_0}$ be the fibration associated to the $\Lambda$-mixed states $\rho_0,\trho_0\in\fr{t}_n$. The Fisher information tensor of $\trho_0$ is $\nabla$-compatible, and it splits as
\begin{equation}
    \fr{F}^{(\trho_0)}= \fr{F}^{(\trho_0)}\vert_{\mathcal{V}} \oplus \fr{F}^{(\rho_0)}.
\end{equation}
\end{proposition}

\begin{proof}
Notice, first, that $\fr{F}^{(\trho_0)} = (\mathbb{DL})^* \mathrm{Tr}\left[\rho \left(\Phi\otimes\Phi\right)\vert_{\trho}\right]$.  Then, we can use Proposition \ref{DLcompatible} and Lemma \ref{Phicompatible} to conclude the argument.
\end{proof}

\begin{corollary}
The Fisher structure $\fr{J}^{(\trho_0)}\colon T\orb{\trho_0} \to T\orb{\trho_0}$ is $\nabla$-compatible w.r.t. the fibration $\orb{\trho_0}\to \orb{\rho_0}$ and $\fr{J}^{(\trho_0)}\vert_F$ is a Fisher structure on the fibres.
\end{corollary}

\section*{Outlook on further research}
The geometric structures described in this paper extend in a natural way the ones arising from geometric quantisation. In particular, the geometry behind the  KKS-symplectic form can be used to reformulate the quantisation procedure of symplectic homogeneous manifolds.

Well-known results in this area can be summarised in the \emph{orbit method} \cite{Kirillov}: a tehcnique to produce (and classify) all irreducible representations of a compact Lie group $G$, combining the theory of coadjoint orbits with geometric quantisation. It follows from a result of Kostant \cite{Kostant} that the KKS symplectic form on a co-adjoint orbit is integral if and only if the orbit passes through the integer lattice in the Weyl chamber in $\fr{t}^*$, the dual of the Cartan subalgebra of $\fr{g}$. This result singles out only certain integral orbits for which geometric quantisation is viable.

The construction we presented in Section \ref{Section:Orbits} can be adapted to general orbits of a compact Lie group but, as mentioned in Remark \ref{Remark:notall}, not all orbits will enjoy a Fisher--K\"ahler structure coming from the symmetric logarithmic derivative. It is reasonable to assume that some will: for the Fisher tensor to exist on an orbit $\orb{\eta_0}$ of $U(n)$ (with $\eta_0\in t^+$) one requires that Equation \eqref{LinearLequation} have solutions on $\nor{\eta_0}\simeq T_{\eta_0}\orb{\eta_0}$.

One interesting direction for further investigation is whether the Fisher form $\fr{W}^{\eta_0}$ on a co-adjoint orbit $\orb{\eta_0}$ gives rise to a non trivial symplectic form for which prequantisation is again viable, especially if the natural $KKS$ symplectic form is not.

More generally, analysing the topological aspects of manifolds that admit a global K\"ahler-Fisher structure might give interesting insights on the geometry of K\"ahler manifolds and homogeneous spaces.
\subsection*{Conflict of Interest Statement} On behalf of all authors, the corresponding author states that there is no conflict of interest.
\subsection*{Data Availability Statement} Data sharing not applicable to this article as no datasets were generated or analysed during the current study.


\begin{thebibliography}{14}
\bibitem{Antonelli}T. Aikou, \emph{Complex Finsler Geometry}, Handbook of Finsler Geometry, Kluwer Academic Publisher, 2003.
\bibitem{B-NG} O. E. Barndorff-Nielsen and R. D. Gill, Fisher information in quantum statistics,
J. Phys.A 33 (2000), 4481.

\bibitem{BFR} M. Bordemann, M. Forger, and H. Römer {\it Homogeneous K\"ahler manifolds: paving the way towards new supersymmetric sigma models}, Comm. Math. Phys. Volume 102, Number 4 (1986), 605-647.
\bibitem{Borel} A. Borel, \emph{Topology of Lie groups and characteristic classes}, Bull. Amer. Math. Soc. 61 (1955), 397-432.
\bibitem{Chat} S. Chaturvedi, E. Ercolessi, G. Marmo, G. Morandi, N. Mukunda, and R. Simon, \emph{Geometric phase for mixed states: A differential geometric approach}, Eur. Phys. J. C 35, 413 (2004). 
\bibitem{CES} I. Contreras, E. Ercolessi and M. Schiavina, {\it On the geometry of mixed states and the quantum information tensor}, J. Math. Phys. 57 (6), (2016)
\bibitem{ES1} E. Ercolessi and M. Schiavina, Geometry of mixed states for a q-bit and the quantum Fisher information tensor, Journal
of Physics A 45 (2012), 365303.

\bibitem{ES2} E. Ercolessi and M. Schiavina, Symmetric logarithmic derivative for general n-level systems and the quantum Fisher
information tensor for three-level systems, Physics Letters A 377 (2013), 1996.
\bibitem {QInf} P. Facchi, R. Kulkarni, V.I. Man’ko, G. Marmo, E. C. G. Sudarshan and F. Ventriglia, Classical and quantum Fisher
information in the geometrical formulation of quantum mechanics, Phys. Lett. A 374 (2010), 4801.
\bibitem{GKM} J. Grabowski, M. Ku`s and G. Marmo, \emph{Geometry of quantum systems: density states and entanglement}, J. Phys. A 381 (2005), 10217.
\bibitem{Guillemin}V. Guillemin, E. Lerman and S. Sternberg, \emph{Symplectic fibrations and multiplicity diagrams}, Cambridde Univ. Press, 1996.
\bibitem {Kirillov} A. A. Kirillov, Lectures on the orbit method, Graduate studies in mathematics 64, (2014)
\bibitem{Kostant} B. Kostant, {\it Quantization and unitary representations}, in Lectures in modern analysis and applications III, Lecture Notes in Math. 170 (1970), Springer Verlag, 87—208.
\bibitem{Kronheimer} P. B. Kronheimer, {\it A Hyper-Kahlerian Structure on Coadjoint Orbits of a Semisimple Complex Group},  Journal of the London Mathematical Society,  Volume  s2-42,  Issue  2,  1  October  1990,  Pages  19320.
\bibitem{Luati1} A. Luati, Maximum Fisher information in mixed state quantum systems, Ann. Stat. 32 (2004), 1770.
\bibitem{Luati2} A. Luati, A Note on Fisher-Helstrom information inequality in purestate models, Indian J. Stat. 70-A (2008), 25.
\bibitem{Neeb} K-H. Neeb, {\it K\"ahler structures and convexity properties of coadjoint orbits}, Forum Mathematicum, Volume 7, Issue 7, Pages 349–384
\bibitem{Reeder} M. Reeder, \emph{On the cohomology of compact Lie groups}, L'Enseignement Math., t. 41, (1995), pages 181-200.
\bibitem{Terr} J. Terrier, {\it Linear complex and almost complex structures}, Proc. AMS 49, 1975

\end{thebibliography}
\end{document}